%% file: efp.tex
\documentclass[twocolumn,pre,showkeys,showpacs,superscriptaddress,preprintnumbers,floatfix]{revtex4}

\usepackage{ifpdf}
\usepackage{hyperref}
\usepackage{dcolumn}
\usepackage{url}
\usepackage{amsmath}
\usepackage{amscd}
\usepackage{amsfonts}
\usepackage{amssymb}
\usepackage{bm}   
\usepackage{bbm}   
\usepackage{verbatim}
\usepackage{stmaryrd}
\usepackage{amsthm}

\ifpdf
  \usepackage[pdftex]{graphicx}         
\else
  \usepackage[dvips]{graphicx}
\fi

\theoremstyle{plain}   \newtheorem{Lem}{Lemma}
\theoremstyle{plain} 	
\theoremstyle{plain} 	
\theoremstyle{plain} 	\newtheorem{Prop}{Proposition}
\theoremstyle{plain} 	
\theoremstyle{plain}	
\theoremstyle{definition}	\newtheorem*{Def}{Definition} 
\theoremstyle{plain}	

\newtheoremstyle{algorithm}{\topsep}{\topsep}%
 {}
 {}
 {\bfseries}
 {.}
 { }
 {\thmname{#1}\thmnumber{ #2}\thmnote{ #3}}

\theoremstyle{algorithm} \newtheorem{Alg}{Algorithm}

\newcommand{\figref}[1]{Fig.~\ref{#1}}

\input{cmechabbrev}

\begin{document}

\title{Enumerating Finitary Processes}

\author{Benjamin D. Johnson}
\email{bjohnson@math.ucdavis.edu}
\affiliation{Complexity Sciences Center}
\affiliation{Mathematics Department}

\author{James P. Crutchfield}
\email{chaos@ucdavis.edu}
\affiliation{Complexity Sciences Center}
\affiliation{Mathematics Department}
\affiliation{Physics Department\\
University of California at Davis,\\
One Shields Avenue, Davis, CA 95616}
\affiliation{Santa Fe Institute\\
1399 Hyde Park Road, Santa Fe, NM 87501}

\author{Christopher J. Ellison}
\email{cellison@cse.ucdavis.edu}
\affiliation{Complexity Sciences Center}
\affiliation{Physics Department\\
University of California at Davis,\\
One Shields Avenue, Davis, CA 95616}

\author{Carl S. McTague}
\email{c.mctague@dpmms.cam.ac.uk}
\affiliation{DPMMS, Centre for Mathematical Sciences,\\
University of Cambridge,\\
Wilberforce Road, Cambridge, CB3 0WB, England}
\affiliation{Santa Fe Institute\\
1399 Hyde Park Road, Santa Fe, NM 87501}

\date{\today}

\bibliographystyle{unsrt}

\begin{abstract}
We show how to efficiently enumerate a class of finite-memory stochastic
processes using the causal representation of \eMs. We characterize \eMs\ in
the language of automata theory and adapt a recent algorithm for generating
accessible deterministic finite automata, pruning this over-large class down to
that of \eMs. As an application, we exactly enumerate topological \eMs\ up to
eight states and six-letter alphabets.
\end{abstract}

\pacs{
02.50.-r  
89.70.+c  
05.45.Tp  
02.50.Ey  
}
\preprint{Santa Fe Institute Working Paper 10-11-027}
\preprint{arxiv.org:1011.0036 [cs.FL]}

\keywords{epsilon-machine, orderly enumeration}

\maketitle


\section{Introduction}

What does the landscape of stochastic processes look like? Some classes of
process---e.g., modeled by Markov chains and Hidden Markov models, finite or
denumerable \cite{Keme65a,Keme76a,Rabi86a,Ephr02a}---are familiar to us since
they have proven so useful as models of randomness in real world systems. Even
if this familiarity belies a now-extensive understanding for particular classes,
it begs the question of the intrinsic organization and diversity found in the
space of all stochastic processes. Randomly selecting a stochastic process,
how often does one find that it saturates the entropy rate? How many distinct
processes are there at a given entropy rate or with a given number of states?
Answers to these and related questions will go some distance to understanding
the richness of stochastic processes and these, in turn, will provide hints as
to what is possible in nature.

Stochastic processes show up in an exceedingly wide range of fields, but they
are not generally analyzed or classified in broad swaths. In an attempt to
address such concerns, we show how to enumerate the class of stochastic
processes that admit the causal representation of finite-state \eMs.

An \eM\ is the minimally complex, maximally predictive representation that
completely captures all of a stochastic process's information storage and
processing properties. The \eM\ representation allows for direct analysis of
the underlying process using only relevant information, and it provides a
framework for comparing different processes through common, measurable
quantities. The literature on computational mechanics
\cite{CompMechMerge}, the area responsible for the theory of \eMs,
provides details about the construction of \eMs\ from process output, proof of
their optimality, various information-theoretic quantities that can be
calculated from the \eM, and more.

Here, we consider stationary stochastic processes over discrete states and
discrete alphabets. Given that each such process can be completely represented
by its \eM, to enumerate all stochastic processes it suffices to enumerate all
\eMs. Even if one restricts to the case of \eMs\ with finitely many states,
this task appears to be extraordinarily difficult. So, as a first step, we
enumerate a subclass of \eMs\ called topological \eMs, which represent a
subclass of all finite-memory processes. In a sequel, we extend the ideas
presented here to more general stochastic processes and their \eMs.

Although we are a long way from mapping the landscape of all stochastic
processes, enumerating a subclass of finite-memory stochastic processes is
useful for a number of reasons. First is basic understanding. One would simply
like to know how many processes there are for a given number of states and
alphabet size. Moreover, if we fix one of these parameters and increase the
other, it is informative to see how the number of distinct processes scales
as well. Second, it allows for a thorough survey of process characteristics.
An example of a such a survey is found in Ref.~\cite{Feld08a}. Third, an
enumerated list of processes can be used to rigorously establish properties for
various kinds of complex systems. A library like this was used in
Refs.~\cite{McTa04a} and~\cite{Crut02a} to prove theorems about pattern
formation in cellular automata. Finally, and rather more generally, one needs
to be able to sample and explore the space of processes in a random or a
systematic way, such as required in Bayesian inference~\cite{Stre07a}.

Starting from an algorithm initially designed to enumerate deterministic finite
automata, we use \eM\ properties as a selection criteria for these automata,
resulting in the set of topological \eMs\ (and the processes they describe) as
a result. Our development of this is organized as follows. First, we briefly
discuss our previous approach to this problem using a different orderly
enumeration algorithm due to Read~\cite{Read78a}, followed by an
overview of the algorithm on which our enumeration scheme is 
based~\cite{Alme06a,Alme07a}. Second, we lay out the machinery of this algorithm, 
reviewing automata theory and computational mechanics. We define the necessary 
concepts as they apply to topological \eM\ generation and enumeration.
Third, we then describe our algorithm, give pseudocode for its implementation,
and prove that it successfully enumerates all topological \eMs. 
Fourth, we present enumeration results as a function of the number of states
and symbols. We discuss, as well, the performance of the new algorithm,
comparing it to our previous algorithm, and explain the improvements.

\section{Related Work}

The enumeration of \eMs\ has not, to our knowledge, been previously explored,
outside of the above-cited works. The enumeration of certain classes of DFAs,
in contrast, has been pursued with varying degrees of success. Of particular
interest, strongly connected and minimal complete finite automata were
separately enumerated in Refs.~\cite{Robi85a} and ~\cite{Naru77a}, respectively. See Ref.~\cite{Doma02a} and references therein for more details on other recent efforts.

Much of the literature on computational mechanics focuses on \eMs\ from the
standpoint of Markov chains and stochastic processes and, therefore, typically
uses the transition matrices as an \eM's representation. Our first approach for
enumerating finitary processes focused on generating all possible transition
matrices and, hence, all \eMs, interpreted as labeled directed graphs.
Read~\cite{Read78a} presented
an orderly generation algorithm that could be used to efficiently generate
certain classes of combinatorial objects. Among the objects that can be
generated are directed and undirected graphs, rooted trees, and tournaments
(interpreted as a special class of directed complete graphs). The essence of
Read's algorithm is that, given the complete list $\mathcal{L}_m$ of graphs 
with $n$ nodes and $m$ edges, we can construct the complete list
$\mathcal{L}_{m+1}$ of graphs with $n$ nodes and $m+1$ edges without having to
run an isomorphism check against each of the already constructed graphs. This
offers a significant speed improvement versus the classical method.

We initially adapted Read's algorithm to generate all edge-labeled multi-digraphs
(with loops). From this extensive list, we then eliminated graphs that were not 
strongly connected and minimal in the sense of finite automata theory. While 
this algorithm was successful, it had three main performance drawbacks: 1) A 
large memory footprint, as $\mathcal{L}_m$ must be stored to generate 
$\mathcal{L}_{m+1}$; 2) an improved, but still extensive, isomorphism check for 
each generated graph---the worst-case scenario requires $n!$ comparisons for
each generated graph; and 
3) generation of a substantially larger class than needed and, as a consequence,
many graphs to eliminate.

Our second approach, and the one presented in detail here, uses a different
representation of \eMs, looking at them as a type of deterministic finite
automata (DFA). The new algorithm suffers from none of the previous method's
problems. Although, it should be noted that this method cannot be used to
enumerate the generalized structures available via Read's algorithm.

In his thesis, Nicaud ~\cite{Nica00a} discussed the enumeration of
``accessible'' DFAs restricted to binary alphabets. These results were then
independently extended to $k$-ary alphabets in Refs.~\cite{Cham05a} and
~\cite{Bass07a}. Recently, Almeida et al.~\cite{Alme06a,Alme07a} developed an algorithm
that generates all possible accessible DFAs with $n$ states and $k$ symbols
using a compact string representation initially discussed in Refs.
~\cite{Alme06a,Alme07a}. They showed that considering the ``skeleton" of these DFAs as $k$-ary trees with $n$ internal nodes guarantees that a DFA's states are all accessible from a start state. From there, they procedurally add edges to the tree in all possible ways to generate all DFAs. As it is possible to generate all such trees, they show that it is possible to generate all accessible DFAs.
They continue on to discuss their enumeration in comparison to the methods of
Refs.~\cite{Cham05a} and ~\cite{Bass07a}, as well as giving a brief commentary
on the percentage of DFAs that are minimal for a given number of states and
symbols.

\section{Automata Representations}

\newcommand{\nullword} {\lambda}

We start with notation and several definitions from automata
theory~\cite{Hopc79} that serve as the basis for the algorithm.

\begin{Def} A \emph{deterministic finite automaton} is a tuple
$\langle Q,\Sigma,\delta,q_0,F \rangle$, where $Q$ is a finite set of states,
$\Sigma$ is a discrete alphabet, $\delta :Q\times\Sigma\to Q$ is the transition
function, $q_0$ is the \emph{start state}, and $F\subseteq Q$ is the set of
\emph{final} (or \emph{accepting}) \emph{states}.
\end{Def}

We extend the transition function in the natural way, with
$\delta(q,\nullword) = q$, for all $q \in Q$, and for
$v,v'\in \Sigma,\ \delta(q,vv') =
\delta(\delta(q,v),v')$.
Here, $\nullword$ denotes the empty word.

With $|Q|=n$ and $|\Sigma| = k$, we take our set of states to be
$Q=\{0,\ldots,n-1\}$ and our alphabet to be $\Sigma=\{0,\ldots,k-1\}$.
When context alone is not clear, states and symbols will be
denoted by $q_i$ and $v_j$, respectively. We will use $F=Q$
(all states are accepting) for our algorithm, although this is not
a general characteristic of DFAs, but is a property of \eMs.

\begin{Def}
A DFA is \emph{complete} if the transition function $\delta$ is
total. That is, for any state $q\in Q$ and symbol $v\in
\Sigma, \delta(q,v)=q'$ for some $q'\in Q$. 
\end{Def}
The DFAs generated by the Almeida et al algorithm may be
incomplete~\cite{Alme06a,Alme07a}. Shortly, we will see this is a necessary condition
for the DFA to be a topological \eM.

\begin{Def}
Two states, $q$ and $q'$, of a DFA are said to be \emph{equivalent} if for all words $w\in\Sigma^*$, $\delta(q,w)\in F$ if and only if $\delta(q',w)\in F$. That is, for every word $w$, following the transitions from $q$ and $q'$ both lead to accepting or nonaccepting states. A DFA is \emph{minimal} if there are no pairwise equivalent states.
\end{Def}
As we take $F=Q$ for \eMs, we can simplify the idea of equivalence somewhat.
Two states of a topological \eM\ are equivalent if the sequences following
each state are the same.
\begin{Def}
A DFA is \emph{accessible} or \emph{initially connected} if for any state
$q\in Q$, there exists a word $w\in \Sigma^*$ such that $\delta(q_0,w)= q$. 
\end{Def}

Simply put, there is a directed path from the initial state to any other state.
The reverse is not necessarily true.

\begin{Def}
A DFA is \emph{strongly connected} if for any two states
$q,q'\in Q$, there is a word $w\in \Sigma^*$ such that
$\delta(q,w)=q'$. Equivalently, for any state $q\in Q$, setting
$q_0=q$ results in the DFA still being accessible.
\end{Def}

\begin{Def}
Two DFAs are \emph{isomorphic} if there is a one-to-one map between the
states that 1) maps accepting and nonaccepting states of one DFA
to the corresponding states of the other, 2) preserves adjacency,
and 3) preserves edge labeling when applied to $\delta$.
\end{Def}

\begin{Def}
A \emph{finite \eM} is a probabilistic finite-state machine with a set of
\emph{causal states} $\{\sigma_0,\ldots,\sigma_{n-1}\}$, a finite alphabet
$\{v_0,\ldots,v_{k-1}\}$, and
transition probabilities on the edges between states, given by a set of
transition matrices $T^{(i)}$, $i\in \{0,\ldots,k-1\}$. Given the current state,
a transition is determined by the output symbol. States are probabilistically
distinguishable, so the \eM\ is minimal.
\end{Def}

An \eM\ has transient and recurrent components, but we only focus on the
recurrent portion, as the transient component can be calculated from the
recurrent. In the following, when we talk about \eMs, we implicitly 
refer to the recurrent states. With this restriction, \eMs\ are also strongly
connected.

Figure~\ref{fig:EvenProcess} gives the \eM\ for the
\emph{Even Process}~\cite{Crut01a}.
The Even Process produces binary sequences in which all blocks of uninterrupted
$1$s are even in length, bounded by $0$s. Furthermore, after each even length
is reached, there is a probability $p$ of breaking the block of $1$s by
inserting a $0$. If a $0$ is inserted, then the same rule applies again.

\begin{figure}[th]
\centering
\includegraphics{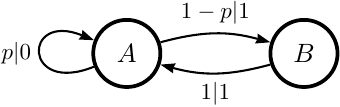}
\caption{Even Process. The transition labels denote the probability
  $p \in (0,1)$ of generating symbol $x$ as $p|x$.
  }
\label{fig:EvenProcess}
\end{figure}

\begin{Def}
A \emph{topological \eM} is an \eM\ where the transition probabilities from
a single state are uniform across all outgoing edges.
\end{Def}

The topological \eM\ for the Even Process is given in 
Fig.~\ref{fig:TopEvenProcess}. We see that the transitions on both edges leaving
state $A$ have probability $1/2$, instead of $p$ and $1-p$ as they were in the 
original Even Process \eM.

\begin{figure}[th]
\centering
\includegraphics{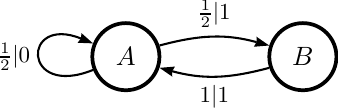}
\caption{Topological \eM\ for the Even Process. Transition probabilities are
  uniform across edges leaving state $A$.
  }
\label{fig:TopEvenProcess}
\end{figure}

Since the transition probabilities are uniform across all edges leaving each
single state, we only need to know their number. As far as the enumeration
algorithm is concerned, we may effectively ignore the
probabilities and focus instead on where the edges go.

This makes clear the name topological \eM: We are only interested in the
topological structure (connectivity or adjacency) as this determines all its
other properties.

One of the key reasons for the success of the algorithm is its compact
representation of DFAs which allows for direct enumeration. Recall that
$|\Sigma|=k$ and suppose that there is a fixed ordering $0,\ldots,n-1$ on the
states $Q$.

\begin{Def}
A DFA's \emph{string} $S = [t_0, t_1, \ldots, t_{nk-1}]$ is an
$nk$-tuple that specifies the terminal state $t_i \in Q$ on each outgoing edge.
The first $k$ entries in the string correspond to the states reached by
following the edges labeled $0,\ldots,k-1$ that start in state $0$. The next
$k$ $t_i$s correspond to the edges that start in state $1$ and so on. Thus,
for each of the $n$ states, there are $k$ specified transitions. If an outgoing
edge does not exist, the corresponding index is marked with $t_i = -1$.
\end{Def}

For clarity, let's consider the topological \eM\ for the Even Process. Let
states $A$ and $B$ be denoted by $0$ and $1$, respectively. The transition
symbols will also be $0$ and $1$, though there is no connection between the
two labelings. As $A$ transitions to $A$ on a $0$ and to $B$ on a $1$, the
terminal states for these two transitions are $0$ and $1$, respectively. $B$
has no outgoing transition on symbol $0$, so that will be denoted $-1$ in the
string, while the transition from $B$ to $A$ on a $1$ will be given by $0$.
Thus, the string representation for the Even Process is $S = [0,1,-1,0]$.

In the definition of a DFA's string, we assumed a fixed ordering on the states.
In general, there are $n!$ ways to label the states and as many strings, so we
need a way to fix a labeling unambiguously. To do this, we label the states in
the order in which they are reached by following edges lexicographically from
state $q_0$. Start with $q_0 \equiv 0$, then follow the edges coming out of
$q_0$ in order: $0,1,\ldots,k-1$. The first state reached that is not state
$0$ is labeled as $1$. The next state that is not $0$ or $1$ becomes state
$2$, and so on. Once the edges $0,\ldots,k-1$ have been explored, the procedure
is repeated, starting from state $1$, then state $2$, and so on---until all the
states have been labeled. Given the initial state $q_0$ of an accessible DFA,
the edges uniquely determine the labeling of all the other states in the DFA.
A proof can be found in Refs.~\cite{Alme06a,Alme07a}. Note that the DFA must be
accessible for this to work, else states will be missed in the labeling process.

\begin{Def}
Given a DFA string $S$, the corresponding \emph{flag} $f=[f_0,f_1,\ldots,f_{n}]$
is an $n+1$ tuple, with $f_0 = -1$, $f_n =nk$, and
$f_i = \min \{j : S_j = i \}$.
That is, $f_i$ is the index of the first occurrence of $i$ in the DFA string
$S$. Note that as the DFA is accessible, $f_i\leq ik-1$.
\end{Def}

The flag for the Even Process shown above is $[-1, 1, 4]$.

\section{Enumeration Algorithm}

To enumerate and generate all topological \eMs, we begin with the Almeida et al
algorithm~\cite{Alme06a,Alme07a} that generates all accessible DFAs, of which
topological \eMs\ are a subclass. We then eliminate
those DFAs that are not \eMs. The following Lemmas help with this
process.

\begin{Lem}
A topological \eM\ with $n$ states has at least $n$ transitions.
\label{lem:MinEdgesLem}
\end{Lem}

\begin{proof}
Assume there are at most $n-1$ transitions. Then there is at least one state
with no outgoing transition. There is no path from this state to any other
state, so this cannot be an \eM, as it is not strongly connected.
\end{proof}

\begin{Lem} A topological \eM\ with $n>1$ states and alphabet size $k$ can
have at most $nk-1$ transitions.
\label{lem:MaxEdgesLem}
\end{Lem}

\begin{proof}
The number of transitions is at most $nk$, as each state can have
at most $k$ transitions. Suppose that an \eM\ has $nk$
transitions. Then every word $w\in\Sigma^*$ is accepting for every
state, so all states are pairwise equivalent. This cannot be an
\eM, since it is not minimal. Thus, there are at most $nk-1$
transitions.
\end{proof}

This establishes our earlier claim that topological \eMs\ are incomplete.

\begin{Lem} A topological \eM\ with $n$ states has $n$ isomorphic string
automata representations.
\label{lem:NumIsomorphsLem}
\end{Lem}

\begin{proof} An \eM\ is strongly connected. In the above
definition of a strongly connected DFA, we gave an equivalent
characterization where any state may serve as $q_0$ and result in
an accessible DFA. As state $q_0$ determines the labeling of the
states, and so the string representations, there are exactly $n$
such representations.
\end{proof}

We now need to determine the canonical representation for a given
topological \eM. Given the $n$ different strings that all
represent the \eM\ equally well, which do we add to our enumerated
list, and how do we know if we already have some isomorphism of an
\eM\ on our list?

A closed-form expression to exactly count the number $B^1_{n,k}$ of incomplete,
accessible DFAs with $n$ states and alphabet size $k$ was developed in 
Refs.~\cite{Alme06a,Alme07a}. A bijection between the integers $0,\ldots,B^1_{n,k}-1$ and the 
DFAs generated by the algorithm was also given. In this way, we can determine 
the $i^\text{th}$ DFA generated by the algorithm and likewise, given an 
arbitrary accessible DFA, we can determine exactly where in the generation
sequence it occurs. This bijection allows us to easily determine
whether an \eM\ is the canonical representation for its
isomorphism class. We denote by $B^1_{n,k}(S)$ the index of the
string representation $S$ in the enumeration process.
Appendix \ref{app:Mapping} gives the details.

\begin{Def} Given the $n$ different string representations of a topological
\eM---$S_1,S_2,\ldots,S_n$---the \emph{canonical representation} $\widehat{S}$ 
is the string with the smallest $B^1_{n,k}$ value. It is the first of the
isomorphisms generated by the enumeration process:
\begin{align*}
 	\widehat{S} \equiv \min_{1 \leq i \leq n} B^1_{n,k}(S_i)~.
\end{align*}
\end{Def}

With this definition of a canonical representation, it is simple to determine
whether a given \eM\ has already been generated: Compute the index
$B^1_{n,k}(S)$ of its representation $S$. Take each state as $q_0$ and compute
the new string representation. If any of the resulting representations has a
lower index than the original, then the given \eM\ is not canonical. So, we
ignore it and generate the next DFA in the enumeration sequence.

To solidify the above ideas, consider the topological \eM\ in \figref{fig:TopEM}.
Note that since transition probabilities are not relevant to the enumeration
process, we omit them entirely and only show the output symbol. Also, note that
we label our states with letters, not numbers, for clarity.

\begin{figure}[th]
\centering
\includegraphics{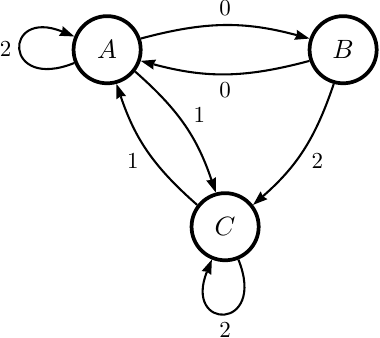}
\caption{Arbitrary topological \eM\ with 3 states over alphabet of size 3.}
\label{fig:TopEM}
\end{figure}

Depending on the choice of $q_0$, there are $3$ different representations of
this \eM:
\begin{enumerate}
\setlength{\topsep}{0pt}
\setlength{\itemsep}{0pt}
\item $q_0 = A$:\\[.075in]
	To determine the state ordering, we follow the edge labeled $0$ and get
	$q_1=B$. We follow the edge labeled $2$ from state $B$ to get $q_2=C$.
	In this way we identify $(A,B,C)$ as $(0,1,2)$ and obtain the string
	representation $S_1=[1,2,0,0,-1,2,-1,0,2]$. From this, we compute that
	$B^1_{n,k}(S_1) = 70791$.
\item $q_0 = B$:\\[.075in]
	We find that $q_1=A$ and $q_2=C$. So, we identify $(A,B,C)$ as $(1,0,2)$
	and determine that $S_2=[1,-1,2,0,2,1,-1,1,2]$. This yields
	$B^1_{n,k}(S_2)=55115$.
\item $q_0 = C$:\\[.075in]
	We identify $(A,B,C) = (1,2,0)$, finding that $S_3=[-1,1,0,2,0,1,1,-1,0]$
	and $B^1_{n,k}(S_3)=18977$.
\end{enumerate}
All three strings are valid representations of the \eM, but the third $S_3$ has
the lowest index ($18977$) in the enumeration sequence, so it is the canonical
representation of the \eM. During the enumeration process the other two
representations would be ignored after it was determined they were noncanonical.

With this information in-hand, we can now provide the pseudocode for our
algorithm. For clarity of discussion, we break the algorithm into two pieces. The first generates accessible DFAs, while the second tests to see if they are topological \eMs.

We only highlight the important aspects of the DFA generation algorithm here.
For a more complete discussion, as well as code for implementation, see
Refs.~\cite{Alme06a,Alme07a}.

\begin{Alg}{DFA Generation}
\label{Alg:DFAgen}
~
\begin{enumerate}
{\bf Input}: Number of states $n$, alphabet size $k$.
\setlength{\topsep}{0pt}
\setlength{\itemsep}{0pt}
\item Generate the flags in reverse lexicographic order.
\item For each flag:
\begin{enumerate}
  \setlength{\topsep}{0pt}
  \setlength{\itemsep}{0pt}
  \item Generate strings with this flag one at a time, in lexicographic order.
  	Each is generated from the previous.
  \item Test the DFA string $S$ to see if it is a canonical topological \eM.
  	(See Algorithm~\ref{Alg:EMTest}.)
  \item If the DFA is canonical, output $B^1_{n,k}(S)$ to the list of
  	topological \eMs.
  \item Move to next flag when all strings have been generated.
\end{enumerate}
\item Terminate after last string for last flag has been generated.

{\bf Output}: The list of indices $\{ B^1_{n,k}(S) \}$ of all topological
\eMs\ for the given $n$ and $k$.
\end{enumerate}
\end{Alg}

\begin{Alg}{Test for topological \eM}
\label{Alg:EMTest}
~
\begin{enumerate}
\setlength{\topsep}{0pt}
\setlength{\itemsep}{0pt}
{\bf Input}: DFA $X$ in string representation $S$ and $B^1_{n,k}(S)$.
\item Reject $X$ unless it has at least $n$ transitions.
\item Reject $X$ if it has $nk$ transitions.
\item For $i=1,\ldots,n-1$:
  \begin{enumerate}
  \setlength{\topsep}{0pt}
  \setlength{\itemsep}{0pt}
  \item Create a new DFA $Y_i$ from DFA $X$ with $q_0=i$.
  \item Reject $X$ if the states of $Y_i$ cannot be labeled by follow edges
  	lexicographically from $q_0$.\\
       (\emph{$X$ is not strongly connected.})
  \item Build string $S_i$ for $Y_i$.
  \item Compute index $B^1_{n,k}(S_i)$.
  \item Reject $X$ if $B^1_{n,k}(S_i) \leq B^1_{n,k}(S)$.\\
        (\emph{$X$ is not canonical.})
  \end{enumerate}
\item Reject $X$ if it is not a minimal DFA. 

{\bf Output}: True or False, whether the input DFA is a canonical
representation of a topological \eM.
\end{enumerate}
\label{alg:TestTopoEM}
\end{Alg}

Note that steps $1$ and $2$ are not formally necessary for the algorithm to
work, as any DFA that fails these will be not strongly connected and nonminimal,
respectively. However, it is quicker to perform these tests than it is to
check for connectedness or minimality, and it is for these reasons that
Lemmas~\ref{lem:MinEdgesLem} and~\ref{lem:MaxEdgesLem} were mentioned.

\begin{Prop}
The above algorithm generates all topological \eMs\ with $n$ states and $k$
symbols.
\label{prop:EnumerationProp}
\end{Prop}

\begin{proof}
It was already shown in Refs.~\cite{Alme06a,Alme07a} that the original algorithm
generates all accessible DFAs with $n$ states and $k$ symbols. We need only
show that our additions result in only topological \eMs\ being generated.

As stated previously, topological \eMs\ are minimal and strongly connected. We
also require a single representative of an isomorphism class. We check that we
only get strongly connected DFAs in step $3(b)$, and we get minimality from step
$4$. Finally, we prune state isomorphisms with the test in step $3(e)$.
\end{proof} 

See Ref.~\cite{Hopc79} for details on the minimization algorithm used here.
Also, note that we are not interested in the minimal DFA itself, only whether
the given DFA \textit{is} minimal. We minimize the automaton and accept it if it
has the same number of states as the original.

Note that the order of the above checks for connectedness, minimality, and
isomorphic redundancy can be changed, but the performance of the algorithm
suffers. The minimization algorithm is the slowest step, so it should
be performed as few times as necessary, which is why it appears last. 

\vspace{-0.1in}
\section{Results}
\vspace{-0.1in}

We ran the algorithm on a range of $n$ and $k$ values. To date, the majority
of work in computational mechanics focused on binary alphabets, so we provide
not only the number $E_{n,2}$ of \eMs\ with a binary alphabet, but
also a breakdown by the number of edges (transitions) for a given number of
states in Table \ref{Table:Binary}.

\begin{table}
\begin{center}
\begin{tabular}{|lcr|r|}
\hline
 States    & Edges & $E_{n,2}$ & $B^1_{n,2}$\\
 \hline
 1     &       & 3   & \\
       &  1    & 2   & \\
       &  2    & 1   & \\
\hline
 2     &       & 7   &  45 \\
       &  2    & 1   & \\
       &  3    & 6   & \\
 
 \hline
 3     &       &  78 & 816 \\
       &  3    &   2 & \\
       &  4    &  22 & \\
       &  5    &  54 & \\
 
 \hline
 4     &       & 1,388 & 20,225  \\
	   &  4    &     3 &  \\
       &  5    &    68 &  \\
 	   &  6    &   403 &  \\
	   &  7    &   914 &  \\
 \hline
 5     &       &  35,186  &  632,700\\
       &  5    &       6  & \\
 	   &  6    &     192  & \\
	   &  7    &   2,228  & \\
	   &  8    &  10,886  & \\
	   &  9    &  21,874  & \\
 \hline
 6     &        & 1,132,613  & 23,836,540\\
	   &   6    &         9  & \\
	   &   7    &       512  & \\
	   &   8    &     9,721  & \\
	   &   9    &    85,974  & \\
       &  10    &   360,071  & \\
       &  11    &   676,326  & \\
 \hline
 7     &        &   43,997,426 & 1,048,592,640 \\
	   &   7    &           18 & \\
	   &   8    &        1,312 & \\
	   &   9    &       37,736 & \\
       &  10    &      526,760 & \\
       &  11    &    3,809,428 & \\
       &  12    &   14,229,762 & \\
       &  13    &   25,392,410 & \\
\hline
8      &        & 1,993,473,480 & 52,696,514,169 \\
       &   8    &            30 & \\
       &   9    &         3,264 & \\
       &  10    &       133,218 & \\
       &  11    &     2,729,336 & \\
       &  12    &    30,477,505 & \\
       &  13    &   190,505,028 & \\
       &  14    &   651,856,885 & \\
       &  15    & 1,117,768,214 & \\
\hline
\end{tabular}
\end{center}
\caption{The number $E_{n,2}$ of binary-alphabet topological \eMs\ as a
  function of the number of states ($n$) and edges ($k$). The number
  $B_{n,2}^1$ of accessible binary DFAs is listed for comparison.
  }
\label{Table:Binary}
\end{table}

Looking at the numbers in the table, we see that the number of \eMs\ increases
quite rapidly, but when compared to the total number $B_{n,2}^1$ of accessible
binary DFAs, the ratios decrease. At $n=3$, 9.6\% of all accessible DFAs were
topological \eMs; while at $n=8$, that ratio was already down to 3.8\%. We
also see that for any given number of states, the majority of \eMs\ have the
maximum number of possible edges. This is not surprising as a DFA is more
likely to be strongly connected with more edges present.

We note that $E_{n,2}$ is now listed on the \emph{On-Line Encyclopedia of
Integer Sequences} as sequence A181554 \cite{OEIS10a}.

We can certainly consider larger alphabets, and Table ~\ref{Table:All} provides
the number $E_{n,k}$ of \eMs\ for a given number of states $n$ and alphabet
size $k$.

Using the data in Table ~\ref{Table:All} we again consider the ratios of
$E_{n,k}/B^1_{n,k}$. Looking at $2$-state machines with an increasing alphabet,
the ratio quickly approaches $1/2$, indicating that almost every
accessible DFA with 2 states is a topological \eM. (Recall that half of all
machines are noncanonical isomorphisms.)

Although data is lacking to make a definitive conclusion, there is also a trend
that the number of \eMs\ increases more rapidly with increasing states (at
large alphabet) than with increasing alphabet size. This agrees with how the
number of accessible DFAs grows given these two conditions, but we need more
data to be sure.

At this point, we need to address two types of overcounting that appear in
Table \ref{Table:All}. The first occurs due to multiple representations of
a process using a larger alphabet. For example, all machines over $l\geq 2$
letters are also machines over $k$ letters for $k>l$. In fact, there are
$\binom{k}{l}$ representations for each $l$-ary machine in the $k$-ary library.
One may be more interested, however, in new structural features and process
characteristics that appear with a larger alphabet rather than the number
of ways we can re-represent machines with smaller alphabets. As such, Table
\ref{Table:FullAlphabet} provides the number $F_{n,k}$ of topological
\eMs\ that employ all $k$ letters. These machines cannot be found for smaller
$k$ and are, thus, ``new" due to the larger alphabet.

The second type of overcounting is due to symbol isomorphism. Certain processes
listed in both Tables \ref{Table:All} and \ref{Table:FullAlphabet} have
multiple representations that are different as \eMs\, but have the
same characteristics---for example, when quantified using information-theoretic
measures of complexity. The Even Process, to take one example, can be considered
as having even-length blocks of $1$s, as depicted in
\figref{fig:TopEvenProcess}, or even-length blocks of $0$s. The measurable
process characteristics are the same for these two processes. We include both
in our list, as the numbers are of interest to those studying finite-state
transducers, as one example.

\begin{table}
\begin{center}
\begin{tabular}{|c|c|c|c|c|c|c|}
\hline
$n \backslash k$ &             2 &         3 &         4 &         5 &       6\\
\hline
			1    &             3 &         7 &        15 &        31 &      63\\
			2    &             7 &       141 &     1,873 &    20,925 & 213,997\\
			3    &            78 &    15,598 & 1,658,606 & 136,146,590 \\
			4    &         1,388 & 3,625,638 \\
			5    &        35,186 \\
			6    &     1,132,613 \\
			7    &    43,997,426 \\
			8    & 1,993,473,480 \\
\end{tabular}
\end{center}
\caption{The number $E_{n,k}$ of topological \eMs\ as a function of number of
  states $n$ and alphabet size $k$.
  }
\label{Table:All}
\end{table}

We also note that Tables \ref{Table:All} and \ref{Table:FullAlphabet} are
incomplete. This is not a shortcoming of the algorithm, but rather a comment
on the exploding number of \eMs. Looking only at the binary alphabet \eMs,
we see that their numbers increase very rapidly.

\begin{table}
\begin{center}
\begin{tabular}{|c|c|c|c|c|c|c|}
\hline
$n\backslash k$ &           2 &         3 &         4 &           5 &       6\\
\hline
			1 &             1 &         1 &         1 &           1 &       1\\
			2 &             7 &       120 &     1,351 &      12,900 & 113,827\\
			3 &            78 &    15,364 & 1,596,682 & 128,008,760 \\
			4 &         1,388 & 3,621,474 \\
			5 &        35,186\\
			6 &     1,132,613\\
			7 &    43,997,426\\
			8 & 1,993,473,480\\
\end{tabular}
\end{center}
\caption{The number $F_{n,k}$ of full-alphabet topological \eMs\ as a
  function of number of states $n$ and alphabet size $k$.
  }
\label{Table:FullAlphabet}
\end{table}

Looking at the generation times for binary alphabet machines in Table
\ref{Table:RunTimes}, we see that the run times increase very rapidly also.
Our estimate for $9$-state binary machines is approximately $35$ CPU days.
Naturally, since they depend on current technology, the absolute times are less
important than the increasing ratios of run times.

\begin{table}[t]
\begin{center}
\begin{tabular}{|c|c|}
\hline
$n$  & time (seconds) \\
\hline
3 & $1.00 \times 10^{-2}$ \\
4 & $1.30 \times 10^{-2}$ \\
5 & $2.75 \times 10^{-1}$ \\
6 & $1.39 \times 10^{1}$  \\
7 & $7.80 \times 10^{2}$  \\
8 & $4.94 \times 10^{4}$  \\
\end{tabular}
\end{center}
\caption{Average run times ($2.4$ GHz Intel Core 2 Duo CPU) to generate all
  binary alphabet topological \eMs\ as a function of the number $n$ of states.
  }
\label{Table:RunTimes}
\end{table}

\vspace{-0.2in}
\section{Applications}
\vspace{-0.1in}

Computational mechanics considers a number of different properties---including
the entropy rate, statistical complexity, and excess entropy---to quantify a
process's ability to store and transform information \cite{CompMechMerge}.
Additionally, there are 
known bounds on a number of these quantities as well as generalizations of
\eMs\ that achieve these bounds; e.g., see the binary \eM\ survey in
Ref.~\cite{Feld08a}. However, little is known about the nonbinary alphabet case
and about other more recently introduced quantities, such as
causal irreversibility and crypticity~\cite{Crut08a}.
A survey of the intrinsic Markov order and the cryptic order for $6$-state
\eMs\ recently appeared in Ref.~\cite{Jame10a}. A series of sequels will
provide additional surveys---all of which depend on the \eM\ libraries we have
shown how to construct.

Beyond this kind of fundamental understanding of the space of stochastic
processes and the genericity of properties, \eM\ enumeration has a range of
practical applications. One often needs to statistically sample
representations of finite-memory stochastic processes and a
library of \eMs\ forms the basis of such sampling schemes. In the
computational mechanics analysis of spatiotemporal patterns in
spatial dynamical systems, \eMs\ play the role of representing
spacetime shift-invariant sets of configurations. The library can
then be used in computer-aided proofs of the domains, particles,
and particle interactions that are often emergent in such systems, as done in
Ref. \cite{Crut02a}.
Finally, in Bayesian statistical inference from finite data,
priors over the space of \eMs\ are updated based on the evidence
the data provides.
Applications along these lines will appear elsewhere.

\section{Conclusion}

Beginning with an algorithm for enumerating and generating accessible
DFAs, we showed how to enumerate all topological \eMs\ based on the fact that
they are strongly connected and minimal DFAs, discounting for isomorphic
redundancies along the way.

There are a number of open problems and extensions to the algorithm and
enumeration procedure to consider. Ideally, we would like to modify this
algorithm, or create an altogether new one, that directly generates
topological \eMs\ without having to generate a larger class of
objects---counted via $B_{n,k}^1$---that we then prune. Failing this, at least
we would like to generate a smaller class of DFAs, perhaps only those that are
strongly connected, so that fewer candidate DFAs need be eliminated.

We would also like to find a closed-form expression for the number of
topological \eMs\ for a given $n$ and $k$. If this is not possible, we would
like reasonable upper bounds on this quantity (better than $B_{n,k}$) and,
perhaps, asymptotic estimates of the number of accessible DFAs that are
actually topological \eMs. Along these lines, we conjecture that for fixed $k$,
$\lim\limits_{n\to\infty} E_{n,k} / B^1_{n,k} = 0$ and, for fixed $n$,
$\lim\limits_{k\to\infty} E_{n,k} / B^1_{n,k} = 1/n$.

\appendix

\section{String-index mapping}
\label{app:Mapping}

Let $S$ be some DFA string representation, and let $f$ be the flag corresponding
to $S$. Then we have $B^1_{n,k}(s) = n_f + n_r$, where:
\begin{align}
n_f & = \sum_{j=1}^{n-1}
  \left[ \prod_{m=0}^{j-1} (m+2)^{f_{m+1}-f_m-1} \right. \nonumber \\
  & ~~~~~~~~~ \times \left. \sum_{l=f_j+1}^{jk-1} \left( (j+1)^{l-f_j}N^1_{j,l} \right) \right]
\label{eqn:nf}
\end{align}
\begin{align}
n_r & = \sum_{j=1}^{n-1}
  \left[ \sum_{l=f_j+1}^{f_{j+1}-1} s_l(j+2)^{f_{j+1}-1-l} \right. \nonumber \\
  & ~~~~~~~ \times \left.
  \left( \prod_{m=j+1}^{n-1}(m+2)^{f_{m+1}-f_m-1} \right) \right] .
\label{eqn:nr}
\end{align}
Equation (\ref{eqn:nf}) calculates the first index that uses the given flag,
and Eq. (\ref{eqn:nr}) calculates the index of the string $S$ among those
DFAs with the given flag.

Eq.~(\ref{eqn:nf}) refers to the number $N^1_{j,l}$ of accessible DFAs whose
string representation has the first occurrence of symbol $j$ occur in position
$l$. It can be defined by a recursive formula and its values stored in a
table for efficient access. For completeness we provide the formulas here,
but for more detail we direct the reader to Refs.~\cite{Alme06a,Alme07a}:
\begin{align*}
N^1_{n-1,j}  & = (n+1)^{nk-1-j} ~, j \in [n-2,(n-1)k-1]
  \displaybreak[0] \\
N^1_{m,mk-1} & =\sum_{i=0}^{k-1}(m+2)^i N^1_{m+1,mk+i} ~, m \in [1,n-2]
  \displaybreak[0] \\
N^1_{m,j}    & = (m+2) N^1_{m,j+1}+N^1_{m+1,j+1} ~, \\
	& ~~~ m \in [1,n-2], j \in [m-1,mk-2] ~.
\end{align*}

\section*{Acknowledgments}

CSM thanks Fred Annexstein for originally pointing out Read's algorithm.
BDJ was partially supported on an NSF VIGRE graduate fellowship.
CSM was partially supported on an NSF REU internship at the Santa
Fe Institute during the summer of 2002.
CJE was partially supported by a DOE GAANN graduate fellowship.
This work was partially supported by the Defense Advanced Research Projects
Agency (DARPA) Physical Intelligence Subcontract No. 9060-000709
and ARO award W911NF-12-1-0234-0. The views,
opinions, and findings contained in this article are those of the authors and
should not be interpreted as representing the official views or policies,
either expressed or implied, of the DARPA or the Department of Defense.

\vspace{-0.1in}

\bibliography{ref,chaos}

\end{document}

%% file: cmechabbrev.tex
\newcommand{\eM}     {$\epsilon$\protect\nobreakdash-machine}
\newcommand{\eMs}    {$\epsilon$\protect\nobreakdash-machines}

